\documentclass{article}
\usepackage{kr}

\usepackage{times}
\usepackage{soul}
\usepackage{url}
\usepackage[hidelinks]{hyperref}
\usepackage[utf8]{inputenc}
\usepackage[small]{caption}
\usepackage{graphicx}
\usepackage{amsmath}
\usepackage{amsthm}
\usepackage{amsfonts}
\usepackage{booktabs}
\usepackage{algorithm}
\usepackage{algorithmic}
\usepackage{natbib}
\usepackage{multirow}
\newtheorem{definition}{Definition}
\newtheorem{proposition}{Proposition}

\title{A Probabilistic Framework for Hierarchical Goal Recognition}

\author{%
Chenyuan Zhang$^1$\footnote{These authors contributed equally.}\and
 Katherine Ip$^2$\footnotemark[1]\and
Hamid Rezatofighi$^1$\and
Buser Say$^1$\and
Mor Vered$^1$
\\
\affiliations
$^1$Monash University\\
$^2$The University of Melbourne
\emails
 \{chenyuan.zhang, hamid.rezatofighi, buser.say, mor.vered\}@monash.edu\\
     katherine.ip@student.unimelb.edu.au
}

\begin{document}

\maketitle

\begin{abstract}
  Goal recognition aims to infer an agent’s goal from observations of its behaviour. In realistic settings, recognition can benefit from exploiting hierarchical task structure and reasoning under uncertainty. Planning-based goal recognition has made substantial progress over the past decade, but to the best of our knowledge no existing approach jointly integrates hierarchical task structure with probabilistic inference. In this paper, we introduce the first planning-based probabilistic framework for hierarchical goal recognition over Hierarchical Task Networks (HTNs). We instantiate the framework by exploiting an HTN planner with a three-stage generative model for likelihood estimation, yielding posterior distributions over goal hypotheses. Empirical results show improved recognition performance over the existing HTN-based recognizer on HTN benchmarks. Overall, the framework lays a foundation for probabilistic goal recognition grounded in hierarchical planning structure, moving goal recognition toward more practical settings.
\end{abstract}

\section{Introduction}

Goal recognition is a well-established problem concerned with identifying an agent’s goal from observations of its behaviour~\citep{baker2009action, geib2009probabilistic, masters2021s}. Humans are particularly adept at this by exploiting the hierarchical structure in activity. For example, in a cooperative cooking task, a participant typically does not track every primitive action of a teammate, but instead organizes behaviour into higher-level activities such as “preparing dessert” or “making soup”, each composed of many primitive actions. This kind of hierarchical reasoning can reduce cognitive load and helps infer goals from partial (and possibly noisy) observations. It is therefore beneficial for goal recognition algorithms to exploit hierarchical structure in a similar way.

Beyond structure, real-world goal recognition also involves substantial uncertainty~\citep{zhang2025probabilistic}.
Observations may be noisy or incomplete, agents may act sub-optimally or stochastically, and early evidence is often insufficient to distinguish between competing hypotheses. This motivates probabilistic formulations that represent uncertainty over candidate hypotheses rather than committing to a single explanation.

A prominent line of work frames goal recognition as planning, following the formulation of \citet{ramirez2009plan}. These methods exploit a domain model and off-the-shelf planning algorithms to generate candidate explanations on demand that are consistent with the observations, avoiding the need to pre-enumerate and explicitly match against a fixed plan library as in earlier library-based recognizers~\citep{geib2009probabilistic, mirsky2017slim}. Despite the broad success of planning-based approaches, to the best of our knowledge there is still no unified recognizer that combines hierarchical structure with probabilistic inference. Most probabilistic planning-based methods perform Bayesian inference over goal hypotheses but assume flat action models and therefore do not exploit hierarchy~\citep{ramirez2010probabilistic, ramirez2011goal, vered2016online, zhang2023goal, masters2019cost}. \citet{holler2018plan} extend goal recognition as planning to hierarchical domains using Hierarchical Task Networks (HTNs), which couple task decompositions with state-based semantics for primitive actions. However, their recognizer remains deterministic, returning only binary accept/reject decisions for each hypothesis rather than posterior beliefs.

In this work, we aim to combine both heirarchical structures and reasoning over uncertainty in a single framework by providing a probabilistic formulation of hierarchical goal recognition over HTNs. We proceed in four steps. First, we formulate HTN goal recognition as Bayesian inference problem. We then describe a three-stage generative model that specifies how a goal hypothesis gives rise to the observations for likelihood estimation. To make inference practical, we develop an approximation that can be implemented with an off-the-shelf HTN planner and a top-$K$ hypothesis selection procedure. Finally, we extend the framework to tolerate exogenous actions. 

Beyond producing posterior for hypotheses, this probabilistic view offers two further advantages over the prior work. First, it enables comparison among competing hypotheses in a way that aligns with standard goal-recognition intuitions (Section~\ref{gr}). Second, it allows hypotheses to retain non-zero probability even when the observations include \emph{exogenous} actions (Section~\ref{exo}), i.e., actions not licensed by the hypothesis’s HTN decomposition, which may arise from noise observation or goal-irrelevant behaviour.

Our contributions are as follows:
\begin{itemize}
  \item We introduce the first planning-based probabilistic framework for hierarchical goal recognition over HTNs, casting HTN goal recognition as Bayesian inference and yielding posterior probabilities over goal hypotheses.

  \item We instantiate the framework with a three-stage generative process that defines a likelihood over observations, and show the framework supports intuitive ranking among competing hypotheses and tolerates exogenous actions.

  \item We develop a practical inference procedure using an off-the-shelf HTN planner, including a likelihood approximation and a \mbox{top-$k$} hypothesis selection strategy. We evaluate our method on two HTN benchmarks, showing improved recognition performance over the existing HTN-based recognizer.
\end{itemize}

\section{Preliminaries}

We first introduce the notation for goal recognition in the propositional, partially ordered HTN setting, following the formulation of \citet{holler2018plan}. To the best of our knowledge, this remains the only HTN-based approach to goal recognition, and thus provides the closest point of reference for our setting. This section is intended primarily as a reference for formal definitions, and readers may safely skip technical details at first reading. The next section introduces our extension of this setting to support Bayesian inference.

\subsubsection{States and Primitive Operators.}
There are two kinds of tasks in an HTN: \emph{compound} (abstract) tasks and \emph{primitive} tasks.
Compound tasks are decomposed into subtasks via methods, whereas primitive tasks require no further decomposition.
For example, in the Kitchen domain, a compound task such as \texttt{make-soup} may decompose into subtasks like \texttt{pick-up(onion)}, \texttt{chop(onion)}, and \texttt{cook}, and a primitive task such as \texttt{pick-up(onion)} is executed as the corresponding action that changes the state.
In this paper, we use the terms \emph{primitive task} and \emph{action} interchangeably.

Let $\mathcal{F}$ be a finite set of ground atoms. A \emph{state} is a set $s\subseteq \mathcal{F}$.
Let $\Sigma$ be a finite set of primitive action symbols. Each operator $\sigma \in \Sigma$ is a tuple
\[
\sigma = \bigl(\mathrm{pre}(\sigma),\,\mathrm{add}(\sigma),\,\mathrm{del}(\sigma)\bigr),
\mathrm{pre}(\sigma), \mathrm{add}(\sigma), \mathrm{del}(\sigma) \subseteq \mathcal{F}.
\]
Operator $\sigma$ is applicable in $s$ iff $\mathrm{pre}(\sigma) \subseteq s$, yielding
$
\gamma(s,\sigma) \,=\, \bigl(s \setminus \mathrm{del}(\sigma)\bigr) \cup \mathrm{add}(\sigma).
$

\subsubsection{Tasks, Methods and Task Networks.}
Let $NT$ be a finite set of compound task symbols; task symbols are drawn from $NT\cup\Sigma$.
A \emph{task network} is a triple $\langle T,\lambda,\prec\rangle$, where
$T$ is a finite set of task identifiers, $\lambda:T\to (NT\cup\Sigma)$ labels each identifier by a task symbol, and
$\prec\subseteq T\times T$ is a strict partial order over task identifiers specifying precedence constraints.

A \emph{method} is a pair
$
m \;=\; \bigl( X_m,\; \langle T_m,\lambda_m,\prec_m\rangle \bigr),
$
where $X_m \in NT$ is the \emph{head} (i.e., the compound task reduced by $m$), and $\langle T_m, \lambda_m, \prec_m \rangle$ is its subtask network (with labels in $NT \cup \Sigma$). Let $\mathcal{M}(X)$ denote the set of methods applicable to $X \in NT$.

\subsubsection{Decomposition and Execution.}
Given a task network $\langle T,\lambda,\prec\rangle$ and an occurrence $u\in T$ with $\lambda(u)\in NT$,
a decomposition selects some $m\in\mathcal{M}(\lambda(u))$ and replaces $u$ by a \emph{fresh copy}
of $\langle T_m,\lambda_m,\prec_m\rangle$, reconnecting predecessor/successor constraints of 
$u$ to the inserted subtasks.
A task network $\langle T,\lambda,\prec\rangle$ is \emph{primitive} iff $\forall u\in T:\ \lambda(u)\in \Sigma$.

Let $N$ and $N'$ denote task networks. We write
$ N \;\Rightarrow\; N'
$
to denote that $N'$ can be obtained from $N$ by one decomposition. We write
$
N \;\Rightarrow^*\; N'
$
iff $N'$ can be obtained from $N$ by zero or more applications of $\Rightarrow$.

A \emph{linearization} of a primitive network
$\langle T,\lambda,\prec\rangle$ is any sequence of $T$ that
respects $\prec$ (i.e., a topological ordering).
Given an initial state $s_0$, a linearization
$\pi = \langle \sigma_1,\dots,\sigma_n\rangle$ is \emph{executable} if
$\sigma_i$ is applicable in $s_{i-1}$ and $s_i=\gamma(s_{i-1},\sigma_i)$ for
all $i$. We also refer to such an executable linearization $\pi$ as an
\emph{execution}; when convenient, we use the term \emph{plan} to refer to
$\pi$ as well.

\subsubsection{HTN Planning.}
An HTN domain is
$\mathcal{D}=\langle \mathcal{F},\Sigma,NT,\mathcal{M}\rangle$, where $\mathcal{M}$ denotes the set of all methods in the domain.
An HTN planning instance is
$\mathcal{I}=\langle \mathcal{D}, s_0, N_0\rangle$, where
$s_0\subseteq\mathcal{F}$ is the initial state and
$N_0=\langle T_0,\lambda_0,\prec_0\rangle$ is the initial task network.
A \emph{solution} or \emph{explanation} to $\mathcal{I}$ is a pair $(N,\pi)$ s.t.:
\begin{itemize}
  \item $N_0 \Rightarrow^* N$ and $N$ is a primitive task network;
  \item $\pi=\langle \sigma_1,\dots,\sigma_n\rangle$ is an execution of $N$ from $s_0$.
\end{itemize}
When there is no ambiguity, we also refer to a solution $(N,\pi)$ itself as an
\emph{execution}, by slight abuse of terminology.

\subsubsection{HTN Goal Recognition.}
Let $\mathcal{D}=\langle \mathcal{F},\Sigma,NT,\mathcal{M}\rangle$ be an HTN domain, $s_0\subseteq\mathcal{F}$ the initial state,
and let $G$ be a finite set of \emph{candidate goal networks}, where each $N^{g}\in G$ (a \emph{hypothesis}) is a task network of the form
$\langle T^g,\lambda^g,\prec^g\rangle$. 
Let $\hat{o}=\langle o_1,\dots,o_m\rangle$ be a sequence of observed primitive symbols with $o_i\in\Sigma$.
An \emph{HTN goal recognition instance} is
$
\mathcal{I}_{\mathrm{GR}}=\langle \mathcal{D},\, s_0,\, G,\, \hat{o}\rangle .
$
For a candidate goal network $N^{g}\in G$, we say that $N^{g}$ \emph{explains} the observations $\hat{o}$ iff
there exists a solution $(N,\pi)$ to the planning instance $\langle \mathcal{D}, s_0, N^{g}\rangle$, such that
$\hat{o}$ is a subsequence of $\pi$ (the observations appear in order, possibly with unobserved actions in between); that is, there exists a strictly increasing embedding
$e:\{1,\dots,m\}\to\{1,\dots,k\}$ with $o_i=\sigma_{e(i)}$ for all $i$.
The \emph{solution set} of $\mathcal{I}_{\mathrm{GR}}$ is defined as 
$
\mathsf{Sol}(\mathcal{I}_{\mathrm{GR}})\;=\;\{\, N^{g}\in G \mid N^{g}\ \text{explains}\ \hat{o}\,\}.
$

In \citet{holler2018plan}, $\mathcal{I}_{\mathrm{GR}}=\langle \mathcal{D},\, s_0,\, G,\, \hat{o}\rangle$ is reduced to an HTN planning problem by
constructing a new instance $\mathcal{I}' = \bigl\langle \mathcal{D}',\ s_0,\ N_0^{\top}\bigr\rangle$ with the transformed domain $\mathcal{D}' = \langle \mathcal{F},\ \Sigma',\ NT',\ \mathcal{M}'\rangle$ by performing the following steps:
\begin{itemize}
\item Apply the \emph{observation–enforcing compilation} $\phi$
      (as in \citet{ramirez2009plan}) to the primitive symbols occurring in \(\hat{o}\),
      obtaining a new operator set \(\Sigma'\) that forces any executable plan to embed \(\hat{o}\) in order.
\item Introduce $X^{\top}\notin NT$ and set \(NT' \;=\; NT \cup \{X^{\top}\}\).
\item For each candidate goal network
      \(N^{g}=\langle T^g,\lambda^g,\prec^g\rangle\in G\),
      add a root method
      $
      m_g = \bigl(X^{\top},\, \langle T^g,\lambda^g,\prec^g\rangle \bigr)$.
      $\mathcal{M}' = \mathcal{M} \cup \{\, m_g \mid N^{g}\in G \,\}
      $. Note all origin methods in $\mathcal{M}$ that mention primitives are also relabeled by the compilation map $\phi:\Sigma\to\Sigma'$.

\item Initial network is defined as \(N_0^{\top}=\langle \{r\},\, \lambda_0,\, \emptyset\rangle\) with \(\lambda_0(r)=X^{\top}\).
\end{itemize}

\citeauthor{holler2018plan} shows that a candidate $N^{g}\in\mathsf{Sol}(\mathcal{I}_{\mathrm{GR}})$ iff there exists a solution $(N,\pi)$ to $\mathcal{I}'$
such that
$
\exists\,N_1\;\; \text{s.t.}\;\; N_0^{\top}\ \Rightarrow_{m_g}\ N_1\ \Rightarrow^*\ N .
$
However, their HTN-based recognizer returns a \emph{single} solution, even though the underlying framework may, in principle, admit many explanations. To recover the complete set of solutions using a planner that produces only one plan per call, we can instead evaluate each candidate goal network independently, at the cost of additional planner calls: for every $N^{g}\in G$, we construct a restricted planning instance $\mathcal{I}'_{g}$ in which $N^{g}$ serves as the initial task network. If $\mathcal{I}'_{g}$ admits a solution, then $N^{g}$ in included in $\mathsf{Sol}(\mathcal{I}_{\mathrm{GR}})$.

\section{Probabilistic Goal Recognition on HTN}


In the formulation above, recognition is a \emph{feasibility} check, yielding a binary accept/reject decision and therefore doesn't provide a ranking over competing goal hypotheses in the solution set. Moreover, it requires every observed action to be explainable by HTN decomposition, and thus does not allow for any exogenous actions that can arise from noise or suboptimal behaviour. For example, consider an agent whose intention is to \texttt{make-soup}. The observations begin with \texttt{pick-up(onion)} and \texttt{chop(onion)}, but the agent then performs \texttt{pick-up(plate)} before resuming the cooking steps. If \texttt{pick-up(plate)} is not licensed by any decomposition of \texttt{make-soup} in the given HTN model, the recognizer rejects the \texttt{make-soup} intention outright, even though the remainder of the behaviour strongly supports it. 
To address these limitations, we cast HTN goal recognition in a \emph{Bayesian} framework, moving beyond binary feasibility decisions to infer a posterior over goal hypotheses. This allows principled ranking and supports robustness to noisy observations and goal-irrelevant actions by assigning non-zero likelihood to imperfectly explained traces.
We begin by formalizing probabilistic goal recognition over HTNs:
\begin{definition}[Probabilistic hierarchy goal recognition (PHGR) instance]\label{def:phgr}
Let $\mathcal{D}=\langle \mathcal{F},\Sigma,NT,\mathcal{M}\rangle$ be an HTN domain and $s_0\subseteq\mathcal{F}$ the initial state.
Let $G \subseteq NT \cup \Sigma$ be a finite set of candidate goals (hypotheses),
and let $\mathrm{Prior}(N^{g})$ denote a prior probability for each hypothesis
$N^{g} \in G$ given the initial state $s_0$.
Observations are a sequence of primitive symbols
$
\hat{o}=\langle o_1,\dots,o_m\rangle, o_i\in\Sigma.
$
A \textbf{probabilistic hierarchy goal recognition (PHGR) instance} is the tuple
$
\mathcal{I}_{\mathrm{PHGR}}=\langle \mathcal{D},\, s_0,\, G,\, \, \hat{o},\, \mathrm{Prior}\rangle.
$
\end{definition}

\begin{definition}[Solution to PHGR instance]
Given $\mathcal{I}_{\mathrm{PHGR}}=\langle \mathcal{D},\, s_0,\, G,\, \, \hat{o},\,  \mathrm{Prior}\rangle$, the \textbf{solution} is the posterior distribution over $N^g\in G$
\[
P(N^{g}\mid \hat{o},s_0) \;=\; \frac{P(\hat{o}\mid N^{g}, s_0)\,\mathrm{Prior}(N^{g})}{\sum_{N^{g'}\in G} P(\hat{o}\mid N^{g'},s_0)\,\mathrm{Prior}(N^{g'})}.
\]
\end{definition}
Given $\mathrm{Prior}$, the remaining challenge is to calculate the likelihood $P(\hat{o}\mid N^{g}, s_0)$.
We therefore derive a generative model from the HTN framework to estimate this likelihood. In the remainder of this section and in the next, we assume there are no exogenous
actions; handling exogenous actions will be  discussed in 
Section~\ref{exo}.

We model the observations $\hat{o}$ generated by a hypothesis $N^g$ as arising from a three-stage generative process:
(i) \emph{decomposition} of $N^g$ to a primitive task network $N$ with probability $P(N\mid N^g)$ (this distribution does not
condition on the initial state $s_0$, as it depends only on the HTN
structure);
(ii) sampling an \emph{executable linearization} $\pi$ of $N$ with probability $P(\pi\mid N, s_0)$; and
(iii) sampling the \emph{observations} $\hat{o}$ from $\pi$ with probability $P(\hat{o}\mid \pi)$.
Equivalently, 
\[
P(\hat{o},\pi,N \mid N^g,s_0)
=
P(N\mid N^g)\,P(\pi\mid N,s_0)\,P(\hat{o}\mid \pi),
\]
which entails that $N$ depends only on $N^g$, $\pi$ depends only on $(N,s_0)$, and $\hat{o}$ depends only on $\pi$.
The likelihood is then obtained by marginalization:
\[
P(\hat{o}\mid N^g,s_0)
=
\sum_{N,\pi} P(\hat{o},\pi,N \mid N^g,s_0).
\]

Before detailing the stages, we make explicit the underlying assumption that the observed trace is generated by the process defined above. This assumption is not unique to our work, as planning-based goal recognition methods routinely rely on behavioral assumptions, such as (approximate) rationality~\citep{ramirez2009plan, masters2019cost, zhang2023goal}. This further motivates a probabilistic formulation: when the real world departs from the assumed model, probabilistic inference can degrade gracefully, instead of collapsing to failure under the hard feasibility tests used by deterministic approaches. In our case, we adopt a generative process that leverages the hierarchical structure of HTNs. The resulting conditional independences follow the HTN generative story: decompositions are determined by the task hierarchy, executability depends on the decomposed primitive network and the initial state, and the observation model depends only on the executed action sequence. This modelling choice is also broadly consistent with how humans interpret others’ behaviour~\citep{geib2009probabilistic,correa2025exploring,zhang2024human}.
As a result, the process serves both as a computational abstraction and as a reasonable approximation in many recognition scenarios. While we use this particular generative process for concreteness, the overall Bayesian framework remains applicable with alternative observation models, provided they define a corresponding likelihood.

\subsection{Stage I: Network Decomposition}
At any decomposition step, a node labeled $X\in NT$ is refined by choosing a method
$m\in\mathcal{M}(X)$. Let $c(m)\in\mathbb{R}_{\ge 0}$ be a cost score for $m$ (e.g., $c(m)=|T_m|$), and fix an inverse temperature $\beta>0$.
We use a Boltzmann (softmax) distribution following \citet{ramirez2010probabilistic}:
\[
P(m\mid X)\;=\;\frac{\exp\!\bigl(-\beta\,c(m)\bigr)}{\sum_{m'\in\mathcal{M}(X)} \exp\!\bigl(-\beta\,c(m')\bigr)}.
\]
Note that this modelling choice biases the agent toward cheaper high-level task decompositions, which aligns with evidence from human studies of hierarchical problem solving~\citep{correa2025exploring}.
For a complete decomposition sequence $N^g \Rightarrow^* N$ that reduces the hypothesis $N^g$
to a primitive network $N=\langle L,\alpha,\prec\rangle$, let $\mathsf{Steps}(N)$ be the multiset of
applied methods in that sequence, and write $X_m$ for the head of method $m$.
The decomposition probability is
\[
P(N\mid N^g)\;=\;\prod_{m\in \mathsf{Steps}(N)} P\bigl(m\mid X_m\bigr).
\]

\subsection{Stage II: Executable Linearization}\label{li}
Given $N=\langle L,\alpha,\prec\rangle$ and initial state $s_0$, we generate an \emph{executable}
sequence $\pi=\langle \sigma_1,\dots,\sigma_k\rangle$ ($k=|L|$) by iteratively choosing uniformly from the
\emph{available} actions. 

At step $t$, let $L'$ be the set of unexecuted actions in $L$. Define the available set
\[
A_t =\Bigl\{\, \ell\in L' \ \Big|\ 
(\not\exists\,\ell'\in L'\text{ s.t. } \ell'\prec \ell)\ \land\
\mathrm{pre}\bigl(\alpha(\ell)\bigr)\subseteq s_{t-1}\Bigr\}.
\]
$A_t$ contains those \emph{ready} tasks: no unexecuted predecessor (i.e., minimal w.r.t.\ $\prec$ among remaining nodes), and the primitive action $\alpha(\ell)$ is applicable in $s_{t-1}$ (its preconditions hold).
We select $\sigma_t$ uniformly from $\{\alpha(\ell): \ell\in A_t\}$ and set $s_t=\gamma(s_{t-1},\sigma_t)$.
Thus the linearization likelihood is
\[
P(\pi\mid N,s_0)
\;=\;
\prod_{t=1}^{k}
\frac{\mathbf{1}\!\left[\sigma_t \in \{\alpha(\ell):\ell\in A_t\}\right]}{\;\left|A_t\right|}\,.
\]

\subsection{Stage III: Observation Model}

Computing the observation likelihood given a plan may appear straightforward; however, it is not when the observer does not know how many actions from the plan have actually been executed when the observations $\hat{o}$ are obtained.
So we marginalize over the unknown number of executed steps:
\[
P(\hat{o}\mid \pi)\;=\;\sum_{t=0}^{|\pi|} 
\underbrace{P(\text{Execute } t \text{ actions}\mid\pi)}_{\text{progress prior}}
\underbrace{P\!\left(\hat{o}\mid \pi_{1:t}\right)}_{\text{alignment likelihood}}.
\]
The implicit assumption in many prior studies that the executed length is known
is a special case of this formulation, obtained by taking the progress prior to
be $1$ at that length and $0$ elsewhere.

When every executed primitive action is observed in order (i.e., full
observability), the alignment likelihood reduces to a pure equality test
\(
P(\hat{o}\mid \pi_{1:t}) \;=\; \mathbf{1}[\hat{o}=\pi_{1:|\hat{o}|}].
\)
Hence only the prefix of length $|\hat{o}|$ can contribute, and we obtain
\[
P(\hat{o}\mid \pi)
\;=\;
P(\text{Execute }|\hat{o}|\text{ actions}\mid\pi)\,
\mathbf{1}[\hat{o}=\pi_{1:|\hat{o}|}],
\]
because the indicator $\mathbf{1}[\hat{o}=\pi_{1:t}]$ is $0$ for all
$t\neq |\hat{o}|$.

An important consequence of this likelihood term is that, when two plans both
perfectly explain the observations under full observability, shorter plans
receive higher likelihood for any reasonable progress prior because longer plans must
distribute their execution probability over more possible execution lengths.
This matches the human preference for shorter completions reported by
\citet{zhang2024human}, but in our case the effect emerges from the likelihood
(via the progress prior) rather than from a hand-crafted non-uniform prior
over hypotheses as in their formulation.

When observations are partially available, the alignment likelihood
$P(\hat{o}\mid \pi_{1:t})$ evaluates order-preserving alignments (possibly
multiple) of $\hat{o}$ to the executed prefix $\pi_{1:t}$, assigning non-zero
probability whenever $\hat{o}$ is a subsequence of $\pi_{1:t}$ and increasing
with better matches.
It can be computed efficiently via a standard dynamic program over monotone
embeddings, with a fixed per-step detection probability.

\subsection{Likelihood $P(\hat{o}\mid N^{g}, s_0)$ Estimation}\label{gr}

Under the three-stage generative process, each hypothesis $N^g$
induces an unnormalized joint distribution
\[
\tilde P(\hat{o},\pi,N \mid N^g, s_0)
\;=\;
P(N\mid N^g)\,P(\pi\mid N,s_0)\,P(\hat{o}\mid \pi)
\]
over primitive networks $N$, executable linearizations $\pi$, and observations $\hat{o}$. By construction of our generative model, it is straightforward to verify that
$\tilde P(\pi,N \mid N^g, s_0) = 0$ iff $(N,\pi)$ is not a valid solution, and $\tilde P(\hat{o},\pi,N \mid N^g, s_0) = 0$ iff
$(N,\pi)$ is invalid or does not embed $\hat{o}$.
The probability is unnormalized in the sense that some choices in
Stage~I and Stage~II do not lead to any executable sequence from $s_0$.
Consequently, the marginal
\(
\sum_{N,\pi} \tilde P(N,\pi \mid N^g, s_0)
\)
over successful executions is not guaranteed to be $1$ (and can be
substantially smaller in some cases).
The normalized likelihood of $\hat{o}$ under $N^g$ would therefore be
\[
P(\hat{o}\mid N^g, s_0)
\;=\;
\frac{\sum_{N,\pi} \tilde P(\hat{o},\pi,N \mid N^g, s_0)}
     {\sum_{N,\pi} \tilde P(N,\pi \mid N^g, s_0)}\,,
\]
but both the numerator (a sum over all successful executions
consistent with $\hat{o}$) and the denominator (a sum over all successful executions)
are intractable to compute exactly.
Instead of enumerating all executions, we approximate each of these sums by a single
representative term as in prior work~\citep{ramirez2010probabilistic}.
Let
\[
(N^{+},\pi^{+})
\;\in\;
\arg\max_{N,\pi}
\tilde P(\hat{o},\pi,N \mid N^g, s_0)
\]
be the most probable pair that explains $\hat{o}$ (\emph{observation-consistent term}), and let
\[
(N^{\mathrm{base}},\pi^{\mathrm{base}})
\;\in\;
\arg\max_{N,\pi}
\tilde P(\pi,N \mid N^g, s_0)
\]
be the most probable pair not necessarily aligning with $\hat{o}$ (\emph{unconstrained term}).

Based on order-of-magnitude approximation, the numerator is dominated by $\tilde P(\hat{o},\pi^{+},N^{+}\mid N^g, s_0)$
and the denominator by $\tilde P(N^{\mathrm{base}},\pi^{\mathrm{base}}\mid N^g, s_0)$, yielding
\[
P(\hat{o}\mid N^g,s_0)
\;\approx\;
\frac{\tilde P(\hat{o},\pi^{+},N^{+}\mid N^g, s_0)}
     {\tilde P(N^{\mathrm{base}},\pi^{\mathrm{base}}\mid N^g, s_0)}\,.
\]
The computation of the denominator is almost identical to that of the
numerator, with the only difference that we omit Stage~III as
$P(N^{\mathrm{base}} \mid N^g)\,P(\pi^{\mathrm{base}} \mid N^{\mathrm{base}}, s_0)$. Intuitively, this ratio becomes large when there exists an observation-consistent execution $(N^{+},\pi^{+})$ that is almost as probable as the unconstrained execution $(N^{\mathrm{base}},\pi^{\mathrm{base}})$ under the HTN and also explains the observations $\hat{o}$ well.

A key consequence of this probabilistic view is that it addresses a fundamental
limitation of the existing HTN-based recognizer~\citep{holler2018plan}.
Consider an example where hypothesis $g_1$ admits two plans $p_1$ and $p_2$ and hypothesis $g_2$ admits a single plan $p_3$, with
$c(p_1) < c(p_2) < c(p_3) $, and the observations can be explained by $p_2$ and $p_3$
but not by $p_1$.
The HTN-based recognizer that minimizes plan cost over all hypotheses would then
return $g_1$ (via $p_2$), because $c(p_2) < c(p_3)$, even though explaining the
observations under $g_1$ requires deviating from its cheaper plan $p_1$.
In our formulation, by contrast, $P(\hat{o}\mid g)$ depends on the
\emph{gap} between the best unconstrained execution and the best
observation-consistent execution for each hypothesis.
In the example, $g_1$ incurs a gap between $p_1$ (its best unconstrained
plan) and $p_2$ (its best explanation of $\hat{o}$), while $g_2$ has no such
penalty because $p_3$ is both its best unconstrained and best
observation-consistent plan.
As a result, our model can assign higher posterior probability to $g_2$ than to
$g_1$, in line with the general intuition in the goal recognition community that the ``less surprising'' hypothesis should be preferred.

\section{Solving PHGR Instances with an HTN Planner}

In the previous section, we formulated PHGR and its Bayesian solution, and
showed how to approximate the likelihood $P(\hat{o}\mid N^{g}, s_0)$ using two
representative HTN executions.
In this section, we discuss how these representative executions can be computed via an off-the-shelf HTN planner, and how to solve a PHGR problem by ranking hypotheses accordingly.

\subsection{Approximating the Most Probable $(N,\pi)$ Pair}\label{approx}

After reducing likelihood estimation to two representative executions, identifying these representatives remains non-trivial.
Ideally, we want to compute the most probable observation-consistent execution $(N^{+},\pi^{+})$ that embeds $\hat{o}$ and the most probable unconstrained execution $(N^{\mathrm{base}},\pi^{\mathrm{base}})$.
In this work, we approximate these pairs by invoking an off-the-shelf HTN planner on the hypothesis network $N^g$ twice: once on the observation-enforcing compilation to embed $\hat{o}$, and once on the original instance.

In both cases, standard HTN planners return a minimum-cost plan. This is only a proxy for the MAP execution under our generative model, and the two objectives can diverge. Below we highlight two situations in which a longer execution can be more probable than a shorter one for the same hypothesis. This mismatch motivates future work on HTN planning that optimizes probability under an explicit generative model rather than plan cost alone.

\subsubsection{During Network Decomposition}

A shorter plan may traverse very ``branchy'' parts of the HTN, which can dilute its probability under Stage~I.
If a hypothesis admits many alternative methods at the top level, and subsequent decompositions again offer many choices, then the probability mass is split across many competing refinement paths. Any particular primitive network reached via this wide region can therefore receive relatively small probability.
By contrast, a longer plan may arise from a narrow, almost deterministic chain of decompositions, concentrating probability along a single refinement path and thereby yielding higher probability overall despite its greater length.

\subsubsection{During Observation Generation}

A second, more subtle effect arises in Stage~III.
Under partial observability, the likelihood $P(\hat{o}\mid \pi)$ aggregates probability over order-preserving alignments of $\hat{o}$ to executed prefixes of the plan.
Longer plans can admit more valid embeddings of $\hat{o}$ within each prefix; when the per-step detection probability is small, this additional alignment mass can outweigh the length penalty induced by the progress prior.



\subsection{\mbox{Top-$k$} Hypotheses Selection}\label{topk}

A natural next step is to apply the above likelihood calculation to each hypothesis
$N^g \in G$, obtain a posterior $P(N^g \mid \hat{o}, s_0)$, and then rank the
hypotheses accordingly.
However, this direct approach faces two major issues.
First, it requires $2|G|$ calls to the HTN planner (one constrained and one
unconstrained run per hypothesis), which becomes prohibitively expensive when
$|G|$ is large.
Second, and more seriously, HTN planners are often incomplete in practice
due to search bounds and time limits~\citep{yousefi2025good}. If the planner times
out on a hypothesis, this does not imply that no plan exists, so we cannot reliably rank it.

To mitigate these problems, we adopt a \mbox{top-$k$} strategy. The choice of 
$K$ can be adapted to the available computational resources and the characteristics of the domain.
We will identify a set of $K$ promising hypotheses and compute
posteriors only for these. 
We do so by iteratively calling the HTN planner on a transformed HTN problem
$\mathcal{I}'$ with a dummy top-level network, as in
~\citet{holler2018plan}.
The first planner call on $\mathcal{I}'$ returns a solution associated with $N^g$.
We then remove the corresponding top-level method $m_g$ to exclude this
solution and call the planner again to obtain a different solution, repeating
this process until we have collected $K$ distinct solutions or the planner
times out.
The resulting solutions correspond to different hypotheses with the shortest
observation-consistent paths.
As discussed previously, they are not guaranteed to be the most probable
hypotheses under our model, but we treat them as an approximation to the
\mbox{top-$k$} posterior hypotheses.

For each of these selected hypotheses, we then run the HTN planner a second
time without observational constraints to obtain the corresponding unconstrained
executions.
The observation-consistent solution and its unconstrained counterpart instantiate the
numerator and denominator respectively in our likelihood approximation, and
we compute normalized posteriors over this restricted set of hypotheses.
This procedure requires at most $2K$ planner calls and limits the impact of
planner incompleteness, since hypotheses outside the selected set never need to
be evaluated.

\section{Handling Exogenous Actions via Task Insertion}\label{exo}

We have developed an algorithm for solving PHGR instances and shown how our probabilistic framework overcomes a key limitation of existing HTN-based approaches by preferring less surprising hypotheses. We now turn to another advantage of the probabilistic view mentioned earlier: it remains well-defined even when observations contain exogenous actions, which is important in goal recognition settings where behaviour may be noisy or partially unrelated to the goal in realistic environments. To the best of our knowledge, this is the first hierarchical framework capable of handling exogenous actions.

An action is \emph{exogenous} with respect to a hypothesis if it cannot be generated by that hypothesis’s HTN decomposition. In HTN planning, \emph{task insertion}~\citep{geier2011decidability,xiao2017hierarchical} provides a standard semantics for accommodating such actions by allowing extra primitive actions to be interleaved with the decomposed plan. This captures noisy or goal-irrelevant observations, and can also model cases where actions outside the hierarchy are required to achieve the goal. We state our analysis below with respect to this task-insertion semantics. We assume that, for a candidate hypothesis, a task-insertion planner returns a solution as a pair $(N,\pi)$, where $N$ is the decomposed primitive task network and $\pi$ is an executable linearization. Under task insertion, $\pi$ may contain additional primitive actions that are not generated by any decomposition of the hypothesis. We treat exactly these additional actions as the set of exogenous actions $L_e$, obtained by taking the actions that appear in $\pi$ but are not part of the primitive network $N$. 

We first describe how to extend the three-stage generative process to account for exogenous actions. We then analyze an idealized setting with a sound and complete task-insertion planner. Finally, we discuss the practical setting in which such planner is not available.


\subsection{Extending the Generative Process for Exogenous Actions}

Under task insertion semantics, exogenous actions correspond to additional primitive actions that be inserted during linearization. Therefore, we only need to modify Stage~II in the generative process.

Given the decomposed primitive network $N=\langle L,\alpha,\prec\rangle$, an initial state $s_0$, and a set of exogenous actions $L_e$, we generate an \emph{executable} sequence $\pi=\langle \sigma_1,\dots,\sigma_k\rangle$ (where $k=|L|+|L_e|$) by iteratively selecting actions uniformly from the \emph{available set}.
At step $t$, let $L'$ denote the set of unexecuted actions in $L \cup L_e$, extending the original definition to include $L_e$. All other aspects remain unchanged from Section~\ref{li}.


\subsection{Guarantees with a Complete Task-insertion Planner}


We first state a basic \emph{posterior-support} property of the \emph{exact} Bayesian model, and then show that the same property is preserved by the \emph{approximate} algorithm proposed in the last section, assuming the underlying task-insertion planner is sound and complete.

\begin{proposition}[Exact model: posterior support]
Given a PHGR instance $\mathcal{I}_{\mathrm{PHGR}}=\langle \mathcal{D},\, s_0,\, G,\, \hat{o},\, \mathrm{Prior}\rangle$
and a hypothesis $N^g\in G$ with $\mathrm{Prior}(N^g)>0$, let
$P(\hat{o}\mid N^g,s_0)$ be defined by the \emph{exact} marginalization under the
three-stage generative model.
Then \( P(N^g \mid \hat{o}, s_0) > 0 \) iff $\exists (N,\pi)$ such that $(N,\pi)$ is a solution to $\langle \mathcal{D}, s_0, N^{g}\rangle$ and $\hat{o}$ is a subsequence of $\pi$.
\end{proposition}

\begin{proof}[Proof Sketch]
    Since $\mathrm{Prior}(N^g)>0$, we have $P(N^g \mid \hat{o}, s_0)>0$ iff
$P(\hat{o}\mid N^g,s_0)>0$.
Under the exact model, $P(\hat{o}\mid N^g,s_0)$ is a normalized sum of nonnegative
terms $\tilde P(\hat{o},\pi,N\mid N^g,s_0)$ over $(N,\pi)$.
By construction of the generative process, $\tilde P(\hat{o},\pi,N\mid N^g,s_0)>0$
holds exactly for valid executable pairs $(N,\pi)$ that embed $\hat{o}$.
Therefore the sum is positive iff at least one such pair exists. 
\end{proof}

\begin{proposition}[Approximate model: posterior support under a sound and complete task-insertion planner]
    Given a PHGR instance $\mathcal{I}_{\mathrm{PHGR}}=\langle \mathcal{D},\, s_0,\, G,\, \hat{o},\, \mathrm{Prior}\rangle$
and a hypothesis $N^g\in G$ with $\mathrm{Prior}(N^g)>0$. With a sound and complete task-insertion planner, the procedure of Section~\ref{approx} is guaranteed to return a positive posterior
$P_{\mathrm{approx}}(N^g\mid \hat{o},s_0)$ iff $\exists (N,\pi)$ such that $(N,\pi)$ is a solution to $\langle \mathcal{D}, s_0, N^{g}\rangle$ and $\hat{o}$ is a subsequence of $\pi$.
\end{proposition}

\begin{proof}[Proof Sketch]
Since $\mathrm{Prior}(N^g)>0$, we have
$P_{\mathrm{approx}}(N^g\mid \hat{o},s_0)>0$ iff $P_{\mathrm{approx}}(\hat{o}\mid N^g,s_0)>0$.
Under Section~\ref{approx}, the estimator is
\[
P_{\mathrm{approx}}(\hat{o}\mid N^g,s_0)
=
\frac{\tilde P(\hat{o},\pi^{+},N^{+}\mid N^g,s_0)}
     {\tilde P(\pi^{\mathrm{base}},N^{\mathrm{base}}\mid N^g,s_0)},
\]
where $(N^{+},\pi^{+})$ is a planner-returned solution to the observation-enforcing compilation (when solvable) and
$(N^{\mathrm{base}},\pi^{\mathrm{base}})$ is a planner-returned solution to the original (unconstrained) instance.

\smallskip
\noindent\emph{($\Rightarrow$)} If there exists a solution $(N,\pi)$ to $\langle \mathcal{D}, s_0, N^{g}\rangle$ with $\hat{o}$ as a subsequence of $\pi$, then the observation-enforcing compilation is solvable. By completeness, the planner returns some observation-consistent solution $(N^{+},\pi^{+})$. By construction of Stages~I--II, every refinement choice and action-selection step has strictly positive probability, and since $\hat{o}\sqsubseteq \pi^{+}$ the observation model yields $P(\hat{o}\mid\pi^{+})>0$; hence the numerator is strictly positive. Moreover, since $\langle \mathcal{D}, s_0, N^{g}\rangle$ is solvable, completeness also guarantees an unconstrained solution $(N^{\mathrm{base}},\pi^{\mathrm{base}})$, and the same argument gives a strictly positive denominator. Therefore $P_{\mathrm{approx}}(\hat{o}\mid N^g,s_0)>0$, and thus $P_{\mathrm{approx}}(N^g\mid \hat{o},s_0)>0$.

\smallskip
\noindent\emph{($\Leftarrow$)} If $P_{\mathrm{approx}}(N^g\mid \hat{o},s_0)>0$, then $P_{\mathrm{approx}}(\hat{o}\mid N^g,s_0)>0$, so the numerator above must be strictly positive. Positivity of $\tilde P(\hat{o},\pi^{+},N^{+}\mid N^g,s_0)$ implies that $(N^{+},\pi^{+})$ is an executable solution returned by the sound planner and that $P(\hat{o}\mid\pi^{+})>0$. By the observation model, $P(\hat{o}\mid\pi^{+})>0$ holds only when $\hat{o}$ is a subsequence of $\pi^{+}$. Hence $(N^{+},\pi^{+})$ is a solution to $\langle \mathcal{D}, s_0, N^{g}\rangle$ that embeds $\hat{o}$, establishing the existence of the required $(N,\pi)$. \end{proof}

Intuitively, for a given hypothesis $N^g$, observing an additional action that is not goal-directed to that hypothesis should decrease the likelihood of the observation sequence. The following proposition formalizes this monotonicity.

\begin{proposition}[Likelihood monotonicity under an additional observation]
    Fix $\mathcal{I}_{\mathrm{PHGR}}=\langle \mathcal{D}, s_0, G, \hat{o}, \mathrm{Prior}\rangle$ and a hypothesis $N^g$ with $\mathrm{Prior}(N^g)>0$.
Let $\hat{o}'$ be obtained from $\hat{o}$ by inserting one additional primitive symbol $a\in\Sigma$ while preserving the relative order of the symbols in $\hat{o}$.
Then
\[
P_{\mathrm{approx}}(\hat{o}'\mid N^g,s_0)\ <\ P_{\mathrm{approx}}(\hat{o}\mid N^g,s_0).
\]
\end{proposition}

\begin{proof}[Proof Sketch]
Under the approximation of Section~\ref{approx},
\[
P_{\mathrm{approx}}(\hat{o}\mid N^g,s_0)
=
\frac{\tilde P(\hat{o},\pi^{+},N^{+}\mid N^g,s_0)}
     {\tilde P(\pi^{\mathrm{base}},N^{\mathrm{base}}\mid N^g,s_0)}.
\]
The denominator is independent of the observation sequence, so it is identical for $\hat{o}$ and $\hat{o}'$.

For the numerator, as the additional observation is realized via task insertion in the execution and does not change the decomposition, the Stage~I factor $P(N^{+}\mid N^g)$ is unchanged.
Stage~II assigns probability by multiplying per-step selection probabilities, hence an inserted action adds at least one factor (and may reduce others by enlarging the available set) and make $P(\pi^{+}\mid N^{+},s_0)$ decreases.
For Stage~III, since $\pi{+}$ is longer, a reasonable progress prior also assigns less probability to the required executed length.
Multiplying the Stage~II and Stage~III effects yields the stated monotonicity. \end{proof}

\subsection{Exogenous Actions without Task Insertion Support}

Although HTN planning with task insertion admits a compilation to standard HTN planning in theory, earlier work reports a lack of implementation of task-insertion planner, and more recent work provides implementation only under some restrictions~\citep{bercher2019survey,pantuuvckova2025parsing}. Here we discuss what can and cannot be done when a task-insertion planner is not available.

If the planner supports task insertion but is incomplete, we can still apply the \mbox{top-$k$} hypothesis selection procedure from Section~\ref{topk} and compute a posterior by normalizing over the resulting subset of hypotheses. However, unlike Proposition~2, no guarantee can be established for this posterior.

If the available planner does \emph{not} support task insertion, then it cannot handle any hypothesis whose unconstrained explanation requires exogenous actions. In this setting, the need for exogenous actions reflects a limitation of the HTN hierarchy’s expressive capacity, rather than an artifact of goal recognition. To circumvent this limitation, we restrict attention to a class of hypotheses $N^g$ for which the unconstrained instance $\langle \mathcal{D}, s_0, N^g\rangle$ admits a solution generated purely by HTN decomposition. Under this restriction, the unconstrained component becomes tractable and can be approximated using our algorithm with the planner.

However, even under this restriction, difficulties arise when the observation sequence $\hat{o}$ contains exogenous actions (e.g., due to noise or suboptimal behaviour). The observation-enforcing compilation constrains any consistent plan to include these actions. If such an action is exogenous with respect to $N^g$ in the original instance, it remains exogenous after compilation; satisfying the constraint therefore requires task insertion. Consequently, the compiled observation-enforcing problem become unsolvable for the planner, and the observation-consistent term cannot be computed.

To summarize, without task insertion support we can evaluate only those hypotheses for which both the unconstrained and observation-consistent instances admit solutions without exogenous actions.


\section{Experiments}
We evaluated our approach on the Kitchen domain and Monroe domain introduced by \citet{holler2018plan} and compared it against their method (\emph{baseline}), which to the best of our knowledge is the only existing goal recognition approach that operates directly on HTN models.
We use the HTN planning system and PGR solver from PANDA
\citep{holler2018generic} to solve all HTN planning problems\footnote{Code and implementation details are available at \url{https://github.com/kat-s-626/probabilistic_hgr}}.
Experiments were conducted on a MacBook Pro with 10 CPUs, 16 GB of memory, and a 3-minute time limit per instance. 

\subsection{Kitchen Domain}
The kitchen domain models the preparation of multi-course meals, which includes 5 starters, 30  main dishes, and 5
 desserts. Each problem instance specifies a unique goal such as preparing a
single main, a starter–main combination, a main–dessert combination, or a full
three-course meal.
This results in more than $1000$ distinct high-level goal combinations,
many of which share long prefixes and partially ordered subtasks. Consequently, the
same initial actions are often consistent with multiple possible hypotheses.


For computing the posterior distribution over candidate goals
$P(N^{g}\mid \hat{o},s_0)$, we impose uniformity across all components of the generative process.
We assume a uniform prior over goals,
$Prior(N^{g}) = 1/|G|$ for all $N^{g} \in G$, 
a uniform progress prior over the number of executed actions
given a plan, i.e.,
$P(\text{execute } t \text{ actions}\mid\pi) = 1/(|\pi|+1)$
for all $t = 0, 1, \ldots, |\pi|$, 
and a uniform method cost function $c(m)$. 
Finally, due to resource constraints, we set $K = 5$ for \mbox{top-$k$} hypotheses selection.



In addition to our three-stage likelihood approximation, we also consider a
simplified Boltzmann distribution over plan cost, following
\citet{ramirez2010probabilistic}, with temperature parameter $\gamma = 1$:
\begin{align*}
    P(\hat{o}\mid N^g,s_0)
&\approx
\frac{\tilde P(\hat{o},\pi^{+},N^{+}\mid N^g, s_0)}
     {\tilde P(N^{\mathrm{base}},\pi^{\mathrm{base}}\mid N^g, s_0)}\\
&\approx
\frac{\exp\!\bigl(-\gamma\,c(\pi^+)\bigr)}
     {\exp\!\bigl(-\gamma\,c(\pi^\mathrm{base})\bigr)}.
\end{align*}
For goal recognition performance, we evaluate \emph{\mbox{top-$k$} accuracy} based on the posterior over hypotheses. This is distinct from the \mbox{top-$k$} procedure used to generate candidate hypotheses. In all experiments, our methods generate at most $K = 5$ candidate goals. Given this fixed candidate set, top-5 accuracy simply measures whether the true goal appears among the five generated hypotheses and therefore does not depend on how the posterior is estimated. In contrast, top-3 and top-1 accuracy depend on the ranking induced by the posterior: \mbox{top-3} checks whether the true goal lies within the three most probable hypotheses, and top-1 corresponds to the single most probable (MAP) hypothesis. Note that our top-1 solution might not be identical to the baseline, since the baseline reports the first returned hypothesis, which is not necessarily the most probable one under our framework. We report results for both formulations, and the full set of scores is shown in Table~\ref{tab:results}.

\begin{table}[ht]
    \centering
    \setlength{\tabcolsep}{6pt}
    \small

    \begin{tabular}{@{}l|c|cccc@{}}
        \toprule
        & & \multicolumn{4}{c}{\textbf{\mbox{top-$k$} Accuracy (\%)}} \\
        \cmidrule(lr){3-6}
        \textbf{Method} & \textbf{$k$}
        & \textbf{20\%} & \textbf{40\%} & \textbf{60\%} & \textbf{80\%} \\
        \midrule
        \multirow{3}{*}{Full-Three-Stage} 
          & 1 & 10.5 & 44.0 & 75.5 & 80.3 \\
          & 3 & 19.8 & 71.8 & 97.4 & 96.3 \\
          & 5 & \textbf{22.3} & \textbf{72.9} & \textbf{99.8} & \textbf{100.0} \\
        \midrule
        \multirow{2}{*}{Full-Simplified} 
          & 1 &  6.0 & 32.6 & 80.9 & 91.3 \\
          & 3 & 12.7 & 56.0 & 95.8 & 98.9 \\
        \midrule
        Baseline-Full
          & - &  6.0 & 32.6 & 80.9 & 91.5 \\
        \midrule
        \multirow{3}{*}{Partial-Three-Stage} 
          & 1 &  4.9 & 21.8 & 59.5 & 66.0 \\
          & 3 & 15.0 & 54.4 & 92.2 & 93.4 \\
          & 5 & \textbf{17.2} & \textbf{65.5} & \textbf{98.4} & \textbf{100.0} \\
        \midrule
        \multirow{2}{*}{Partial-Simplified} 
          & 1 &  3.9 & 29.3 & 75.3 & 86.2 \\
          & 3 & 11.2 & 52.4 & 94.3 & 99.0 \\
        \midrule
        Baseline-Partial
          & - &  4.0 & 22.0 & 73.0 & 86.0 \\
        \midrule
        \multicolumn{2}{l|}{Avg.\ valid hypotheses (Full)} 
          & 4.9 & 4.7 & 3.2 & 3.1 \\
        \multicolumn{2}{l|}{Avg.\ valid hypotheses (Partial)} 
          & 5.0 & 4.6 & 3.0 & 2.8 \\
        \bottomrule
    \end{tabular}
    \caption{\mbox{top-$k$} accuracy. For \mbox{top-5} accuracy, Full-Simplified and Partial-Simplified 
    coincide with Full-Three-Stage and Partial-Three-Stage, respectively. 
    The average number of valid hypotheses is identical within each setting and reported in the last two rows.}
    \label{tab:results}
\end{table}

\subsubsection{Full Observation}

\begin{figure}[ht]
    \centering
    \includegraphics[width=\columnwidth]{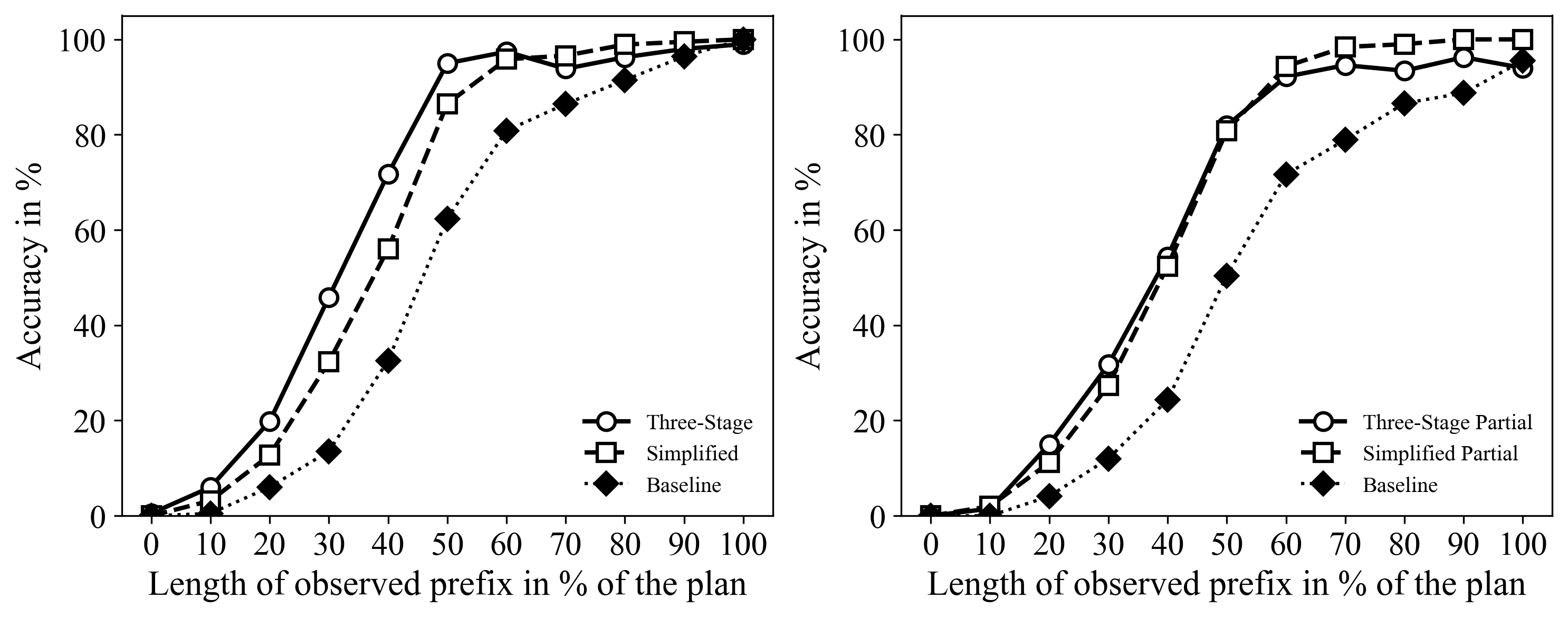}
    \caption{\mbox{top-3} accuracy for posterior estimations and baseline as a function of observation ratio: (a) full observation; (b) partial observation.}
    \label{fig:topk-full}
\end{figure}

We first consider the fully observed setting.
As shown in Table~\ref{tab:results}, both probabilistic models consistently
outperform the baseline in \mbox{top-3} and \mbox{top-5} accuracy across all observation
ratios.
This improved accuracy comes at a moderate computational cost: the
probabilistic models take on average 23.9 seconds per instance, compared to
5.1 seconds for the baseline.
Given the substantial gains in \mbox{top-$k$} accuracy, this additional computation is
well justified for applications that require reliable goal recognition.

For shorter observation prefixes (20\% and 40\%), the three-stage estimator not
only achieves higher top-1 accuracy than the baseline, but also attains higher
\mbox{top-3} and \mbox{top-5} accuracy than the simplified estimator
(Fig.~\ref{fig:topk-full}(a)).

Note that the benchmark traces reflect cost-optimal behaviour induced by the underlying planners used to generate them. Our generative process does not assume strict cost optimality, which introduces a mild model–data mismatch that, if anything, advantages the cost-based baseline and the simplified probabilistic model. Despite this bias, our framework still yields better recognition performance in the early stages, which indicates its strong robustness under uncertainty. Once the observation ratio reaches 60\% and above, the simplified model that
matches the actor’s underlying behaviour begins to overtake the three-stage
estimator.

Finally, the average number of valid hypotheses decreases as more observations
are revealed (i.e., more constraints are imposed) but remains well above
1 (3.1) in the Kitchen domain, reinforcing the need for a
probabilistic treatment rather than committing to a single hypothesis.

\subsubsection{Partial Observation}

We also evaluate our algorithm in a partial-observation setting.
Following \citet{holler2018plan}, partial observation sequences are generated by
randomly removing 20\% of the actions from the full observation trace.

In the partially observed setting, we see qualitatively similar trends to the
fully observed case.
The probabilistic models still offer strong improvements in \mbox{top-$k$} accuracy
over the baseline, and the three-stage estimator remains competitive when
observations are very limited (20\%).
As the observation ratio increases, the advantage of the simplified
model with the correct underlying behaviour emerges earlier, as shown in Table~\ref{tab:results} and
Fig.~\ref{fig:topk-full}(b).
This earlier crossover might attribute to the higher variance introduced by
partial observability, which reduces the robustness margin of the misspecified
three-stage model.

\subsection{Monroe Domain}
To assess generalizability beyond the Kitchen domain, we additionally evaluate on the \textsc{Monroe} domain introduced by \citet{blaylock2005recognizing}, a standard benchmark for plan recognition. The dataset consists of automatically generated plan sessions produced with a modified version of SHOP2. Goals are sampled stochastically, corresponding initial states are generated, and a valid plan is then synthesized for each session. The domain contains 10 lifted goals, 46 lifted methods, and 30 lifted actions, with an average plan length of 9.6 actions per session.
Some problem instances in the domain take more than 500 seconds to search for the goal across all observation prefixes for the baseline algorithm. To ensure computational feasibility, we filtered out these instances and tested our algorithm on the remaining problems.

The largest performance differences appear in the early-observation regime (0--20\% observed), where our approach achieves substantially higher \mbox{top-3} accuracy than the baseline (65\% vs.\ 30\% at 10\% observation; 90\% vs.\ 65\% at 20\% observation). Both methods reach 100\% accuracy once 30\% of the plan is observed, reflecting that most goal hypotheses diverge after only a few initial actions and thus share little prefix structure in the Monroe domain.



\subsection{Exogenous Actions}

As we do not have access to an HTN planner implementation that supports task insertion, we
cannot directly instantiate the ideal setting assumed in our theoretical
analysis.
Instead, we perform a sanity check to highlight the advantage of our
probabilistic approach over the baseline.

To this end, we modify a single Kitchen instance to probe how the
baseline of \citet{holler2018plan} behaves in the presence of exogenous
actions.
In the original instance, the ground truth hypothesis is to prepare spaghetti
carbonara, with goal network
\(\{\textit{makeNoodles}(\texttt{spaghetti},\texttt{p1}),
  \textit{makeCarbonara}\}\).
Under the clean observation trace, the baseline correctly identifies this goal
network.

We then inject a single exogenous action
\(\textit{add}(\texttt{milk},\texttt{p2})\), while keeping all other actions
unchanged.
On this modified trace, the recogniser no longer returns the original goal
set.
Instead, it infers the goal network
\(\textit{makeNoodles}(\texttt{spaghetti},\texttt{p1}),
\textit{makeVanillaPudding}(\texttt{p2})\), \(\textit{makeCarbonara},\)
and produces a plan that interleaves carbonara and pudding preparation.
A single exogenous action is thus “explained’’ by adding an unintended task.

In a second variant, we instead insert
\(\textit{add}(\texttt{spaghetti},\texttt{p3})\) at the corresponding point
in the trace.
Here, the observation-enforcing compilation becomes unsolvable: no HTN plan
achieves the original goals while containing the extra \(\texttt{p3}\) action
at the specified position.
Given the assumption that all observed actions must be generated by the HTN
domain and that the planner cannot perform task insertion, this single
exogenous action causes the recogniser to fail outright.

These two simple modifications highlight a fragility of the existing HTN-based recognizer in the presence of exogenous actions.
Either the system is forced to add additional unintended tasks to
explain the extra actions, or the compiled problem becomes unsolvable under
a planner without task insertion.
In our framework, exogenous actions are handled in Stage~II, and the ground truth hypothesis still receives
non-zero posterior probability.

\section{Related Work}


The main backbone of this paper, namely compiling HTN goal recognition into HTN planning, comes from \citet{holler2018plan}. Their compilation has also been
used in real world settings~\citep{jamakatel2023towards}. This approach can be seen as an
HTN analogue of earlier work on plan recognition as planning for classical
domains~\citep{ramirez2009plan}.
However, these methods are deterministic and can only check whether
hypotheses are compatible with the observations without ranking them. Moreover, as we discuss in Section~3.4, \citet{holler2018plan} also exhibits unexpected behaviour that the system does not prefer less surprising plans
when several hypotheses can explain the observations.

Another important line of work approaches goal recognition from a probabilistic
perspective, going back more than 30 years~\citep{charniak1993bayesian}.
In the 2010s, both the cognitive science community~\citep{baker2009action}
and the planning community~\citep{ramirez2010probabilistic} proposed Bayesian
formulations of goal recognition, estimating likelihoods via plan costs.
This cost-based idea has been widely adopted in subsequent work
\citep{vered2017heuristic,kaminka2018plan,masters2019cost, fitzpatrick2021behaviour, zhang2023goal}.
These approaches are typically defined over non-hierarchical models.
In contrast, our framework provides a more expressive probabilistic
model over a hierarchical structure.

A third strand of work performs goal recognition using \emph{plan libraries}.
These systems reason over a hierarchical library of plans and maintain probabilistic beliefs over goals, as in PHATT~\citep{geib2009probabilistic} and SLIM~\citep{mirsky2017slim}. They provide an explicit probabilistic treatment of hierarchical structure, and their generative views of plan execution informed our own probabilistic formulation. 
However, this line of work has received less attention in recent years. Plan-library methods typically require explicit or implicit enumeration of library plans, and lack full \emph{state-based} action semantics (preconditions and effects), limiting their applicability in many domains. Finally, like the deterministic HTN recognizer of \citet{holler2018plan}, these approaches do not account for \emph{exogenous} actions in the observation trace.

Our work can be viewed as bringing these lines of work together in a
coherent way.
We build a Bayesian inference framework that adopts the generative
probabilistic view of plan-library approaches, and we utilize the observation-enforcing compilation
of \citet{holler2018plan} to approximate the inference problem, while
explicitly modeling exogenous actions within this framework and providing
theoretical guarantees for this setting.

\section{Discussion and Conclusion}
We proposed the first planning-based probabilistic framework for hierarchical goal recognition that
performs Bayesian inference over HTN structures. Our three-stage generative process and likelihood
approximation (via an unconstrained exeuction and an observation-consistent execution)
yields a practical \mbox{top-$k$} hypothesis selection procedure that can be
implemented with an off-the-shelf HTN planner.

With a sound and complete HTN planner that supports task insertion, our
framework provide posterior support for both the exact probabilistic model and our approximation algorithm. Furthermore, the
approximated likelihood is monotonic in the number of exogenous actions: for a
fixed hypothesis it decreases as additional exogenous actions are added to the
observations. Empirical results on the Kitchen and Monroe domains show that the probabilistic framework substantially improves goal recognition performance over the existing HTN-based recognizer under both full and partial observations, at the cost of a moderate constant-factor increase in runtime (i,e., from 5.1 to 23.9 seconds per instance on average in the Kitchen domain, top-5 setting). Despite the mismatch between the generative process (which
encodes our assumptions on observation generation) and the cost-optimal behaviour in the
benchmark, the framework still performs well. The sanity check with injected actions further suggests that the framework is more tolerant to exogenous behaviour than the prior recognizer.

Taken together, our work provides the first planning-based template for integrating hierarchical structure with probabilistic inference, moving goal recognition toward more practical and expressive settings. Promising directions for future work include: (i) developing HTN planners that directly optimize execution probability under the assumed generative process (rather than plan cost), (ii) developing general-purpose HTN planners with full support for task insertion, and (iii) incorporating insights from human hierarchical planning to refine our generative model and better recognize human behaviour.

\section*{Acknowledgments}
This work is supported by the DARPA Assured Neuro Symbolic Learning and Reasoning (ANSR) program under award number FA8750-23-2-1016.

\section*{AI Declaration}

The authors used AI tools to polish the writing of the manuscript and to assist with debugging the code. These tools were used only for editorial and implementation support; all research ideas, methods, experiments, analyses, and conclusions were developed and verified by the authors.

\bibliographystyle{kr}
\bibliography{kr-sample}

@inproceedings{blaylock2005recognizing,
  author    = {Blaylock, Nate and Allen, James},
  title     = {Recognizing Instantiated Goals Using Statistical Methods},
  booktitle = {IJCAI Workshop on Modeling Others from Observations (MOO-2005)},
  pages     = {79--86},
  year      = {2005}
}

@inproceedings{pantuuvckova2025parsing,
  author    = {Pant{\r{u}}{\v{c}}kov{\'a}, Krist{\'y}na and Bart{\'a}k, Roman},
  title     = {Parsing-Based Planner for Totally Ordered {HTN} Planning with Task Insertion},
  booktitle = {2025 IEEE 37th International Conference on Tools with Artificial Intelligence (ICTAI)},
  pages     = {483--490},
  year      = {2025},
  publisher = {IEEE}
}

@inproceedings{ramirez2011goal,
  author    = {Ram{\'i}rez, Miquel and Geffner, Hector},
  title     = {Goal Recognition over {POMDP}s: Inferring the Intention of a {POMDP} Agent},
  booktitle = {Proceedings of the Twenty-Second International Joint Conference on Artificial Intelligence (IJCAI)},
  pages     = {2009--2014},
  year      = {2011}
}

@inproceedings{kaminka2018plan,
  author    = {Kaminka, Gal A. and Vered, Mor and Agmon, Noa},
  title     = {Plan Recognition in Continuous Domains},
  booktitle = {Proceedings of the AAAI Conference on Artificial Intelligence},
  volume    = {32},
  number    = {1},
  year      = {2018}
}

@article{fitzpatrick2021behaviour,
  author    = {Fitzpatrick, Grady and Lipovetzky, Nir and Papasimeon, Michael and Ram{\'i}rez, Miquel and Vered, Mor},
  title     = {Behaviour Recognition with Kinodynamic Planning over Continuous Domains},
  journal   = {Frontiers in Artificial Intelligence},
  volume    = {4},
  pages     = {717003},
  year      = {2021},
  publisher = {Frontiers Media SA}
}

@article{correa2025exploring,
  author    = {Correa, Carlos G. and Sanborn, Sophia and Ho, Mark K. and Callaway, Frederick and Daw, Nathaniel D. and Griffiths, Thomas L.},
  title     = {Exploring the Hierarchical Structure of Human Plans via Program Generation},
  journal   = {Cognition},
  volume    = {255},
  pages     = {105990},
  year      = {2025},
  publisher = {Elsevier}
}

@inproceedings{holler2018plan,
  author    = {H{\"o}ller, Daniel and Behnke, Gregor and Bercher, Pascal and Biundo, Susanne},
  title     = {Plan and Goal Recognition as {HTN} Planning},
  booktitle = {2018 IEEE 30th International Conference on Tools with Artificial Intelligence (ICTAI)},
  pages     = {466--473},
  year      = {2018},
  publisher = {IEEE}
}

@inproceedings{ramirez2010probabilistic,
  author    = {Ram{\'i}rez, Miquel and Geffner, Hector},
  title     = {Probabilistic Plan Recognition Using Off-the-Shelf Classical Planners},
  booktitle = {Proceedings of the AAAI Conference on Artificial Intelligence},
  volume    = {24},
  number    = {1},
  pages     = {1121--1126},
  year      = {2010}
}

@inproceedings{ramirez2009plan,
  author    = {Ram{\'i}rez, Miquel and Geffner, Hector},
  title     = {Plan Recognition as Planning},
  booktitle = {Proceedings of the 21st International Joint Conference on Artificial Intelligence (IJCAI)},
  pages     = {1778--1783},
  year      = {2009}
}

@article{geib2009probabilistic,
  author    = {Geib, Christopher W. and Goldman, Robert P.},
  title     = {A Probabilistic Plan Recognition Algorithm Based on Plan Tree Grammars},
  journal   = {Artificial Intelligence},
  volume    = {173},
  number    = {11},
  pages     = {1101--1132},
  year      = {2009},
  publisher = {Elsevier}
}

@inproceedings{yousefi2025good,
  author    = {Yousefi, Mohammad and Schmautz, Mario and Haslum, Patrik and Bercher, Pascal},
  title     = {How Good Is Perfect? On the Incompleteness of {A*} for Total-Order {HTN} Planning},
  booktitle = {Proceedings of the International Conference on Automated Planning and Scheduling},
  volume    = {35},
  number    = {1},
  pages     = {112--120},
  year      = {2025},
  doi       = {10.1609/icaps.v35i1.36107}
}

@inproceedings{bercher2019survey,
  author    = {Bercher, Pascal and Alford, Ron and H{\"o}ller, Daniel},
  title     = {A Survey on Hierarchical Planning---One Abstract Idea, Many Concrete Realizations},
  booktitle = {Proceedings of the Twenty-Eighth International Joint Conference on Artificial Intelligence (IJCAI)},
  pages     = {6267--6275},
  year      = {2019}
}

@inproceedings{geier2011decidability,
  author    = {Geier, Thomas and Bercher, Pascal},
  title     = {On the Decidability of {HTN} Planning with Task Insertion},
  booktitle = {Proceedings of the Twenty-Second International Joint Conference on Artificial Intelligence (IJCAI)},
  volume    = {22},
  number    = {3},
  pages     = {1955},
  year      = {2011}
}

@inproceedings{xiao2017hierarchical,
  author    = {Xiao, Zhanhao and Herzig, Andreas and Perrussel, Laurent and Wan, Hai and Su, Xiaoheng},
  title     = {Hierarchical Task Network Planning with Task Insertion and State Constraints},
  booktitle = {Proceedings of the Twenty-Sixth International Joint Conference on Artificial Intelligence (IJCAI 2017)},
  pages     = {4463--4469},
  year      = {2017}
}

@inproceedings{zhang2024human,
  author    = {Zhang, Chenyuan and Kemp, Charles and Lipovetzky, Nir},
  title     = {Human Goal Recognition as Bayesian Inference: Investigating the Impact of Actions, Timing, and Goal Solvability},
  booktitle = {Proceedings of the 23rd International Conference on Autonomous Agents and Multiagent Systems},
  pages     = {2066--2074},
  year      = {2024}
}

@inproceedings{zhang2025probabilistic,
  author    = {Zhang, Chenyuan and Cardenas, Cristian Rojas and Rezatofighi, Hamid and Vered, Mor and Say, Buser},
  title     = {Probabilistic Active Goal Recognition},
  booktitle = {Proceedings of the International Conference on Principles of Knowledge Representation and Reasoning},
  volume    = {22},
  number    = {1},
  pages     = {880--890},
  year      = {2025}
}

@inproceedings{vered2016online,
  author    = {Vered, Mor and Kaminka, Gal A. and Biham, Sivan},
  title     = {Online Goal Recognition through Mirroring: Humans and Agents},
  booktitle = {Annual Conference on Advances in Cognitive Systems},
  year      = {2016},
  publisher = {Cognitive Systems Foundation}
}

@article{baker2009action,
  author    = {Baker, Chris L. and Saxe, Rebecca and Tenenbaum, Joshua B.},
  title     = {Action Understanding as Inverse Planning},
  journal   = {Cognition},
  volume    = {113},
  number    = {3},
  pages     = {329--349},
  year      = {2009},
  publisher = {Elsevier}
}

@inproceedings{mirsky2017slim,
  title={SLIM: semi-lazy inference mechanism for plan recognition},
  author={Mirsky, Reuth and Gal, Ya'akov},
  booktitle={Proceedings of the Twenty-Fifth International Joint Conference on Artificial Intelligence},
  pages={394--400},
  year={2016}
}

@article{masters2019cost,
  author  = {Masters, Peta and Sardina, Sebastian},
  title   = {Cost-Based Goal Recognition in Navigational Domains},
  journal = {Journal of Artificial Intelligence Research},
  volume  = {64},
  pages   = {197--242},
  year    = {2019}
}

@inproceedings{zhang2023goal,
  author    = {Zhang, Chenyuan and Kemp, Charles and Lipovetzky, Nir},
  title     = {Goal Recognition with Timing Information},
  booktitle = {Proceedings of the International Conference on Automated Planning and Scheduling},
  volume    = {33},
  pages     = {443--451},
  year      = {2023}
}

@article{charniak1993bayesian,
  author    = {Charniak, Eugene and Goldman, Robert P.},
  title     = {A Bayesian Model of Plan Recognition},
  journal   = {Artificial Intelligence},
  volume    = {64},
  number    = {1},
  pages     = {53--79},
  year      = {1993},
  publisher = {Elsevier}
}

@inproceedings{masters2021s,
  author    = {Masters, Peta and Vered, Mor},
  title     = {What's the Context? Implicit and Explicit Assumptions in Model-Based Goal Recognition},
  booktitle = {Proceedings of the Thirtieth International Joint Conference on Artificial Intelligence (IJCAI 2021)},
  pages     = {4516--4523},
  year      = {2021},
  publisher = {Association for the Advancement of Artificial Intelligence}
}

@inproceedings{jamakatel2023towards,
  author    = {Jamakatel, Prakash and Bercher, Pascal and Schulte, Axel and Kiam, Jane Jean},
  title     = {Towards Intelligent Companion Systems in General Aviation Using Hierarchical Plan and Goal Recognition},
  booktitle = {Proceedings of the 11th International Conference on Human-Agent Interaction},
  pages     = {229--237},
  year      = {2023}
}

@inproceedings{vered2017heuristic,
  title={Heuristic online goal recognition in continuous domains},
  author={Vered, Mor and Kaminka, Gal A},
  booktitle={International Joint Conference on Artificial Intelligence 2017},
  pages={4447--4454},
  year={2017},
  organization={Association for the Advancement of Artificial Intelligence (AAAI)}
}

@inproceedings{holler2018generic,
  author    = {H{\"o}ller, Daniel and Bercher, Pascal and Behnke, Gregor and Biundo, Susanne},
  title     = {A Generic Method to Guide {HTN} Progression Search with Classical Heuristics},
  booktitle = {Proceedings of the International Conference on Automated Planning and Scheduling},
  volume    = {28},
  pages     = {114--122},
  year      = {2018}
}

\end{document}